\documentclass[a4paper,12pt]{article}
\usepackage{amsthm,amsmath,amssymb,latexsym,bm}

\newtheorem{thm}{Theorem}[section]

\newtheorem{lem}[thm]{Lemma}

\newtheorem{rem}[thm]{Remark}
\numberwithin{equation}{section}

\newcommand{\vep}{\varepsilon}

\newcommand{\wh}{\widehat}

\newcommand{\GradPVI}{
\begin{figure}
\begin{center}
	\begin{picture}(80,70)
		\put(13,8){\circle*{6}}\put(0,5){\small$0$}
		\put(15,10){\line(1,1){20}}
		\put(13,56){\circle*{6}}\put(0,53){\small$1$}
		\put(15,54){\line(1,-1){20}}
		\put(37,32){\circle{6}}\put(34,17){\small$2$}
		\put(39,30){\line(1,-1){20}}
		\put(61,8){\circle*{6}}\put(67,5){\small$4$}
		\put(39,34){\line(1,1){20}}
		\put(61,56){\circle*{6}}\put(67,53){\small$3$}
	\end{picture}
\caption{Gradation of $\mathfrak{g}(D^{(1)}_4)$ of type $(1,1,0,1,1)$}
\end{center}
\end{figure}
}
\newcommand{\GradSas}{
\begin{figure}
\begin{center}
	\begin{picture}(170,70)
		\put(13,8){\circle*{6}}\put(0,5){\small$0$}
		\put(15,10){\line(1,1){20}}
		\put(13,56){\circle*{6}}\put(0,53){\small$1$}
		\put(15,54){\line(1,-1){20}}
		\put(37,32){\circle{6}}\put(34,17){\small$2$}
		\put(40,32){\line(1,0){20}}
		\put(63,32){\circle*{6}}\put(60,17){\small$3$}
		\put(66,32){\line(1,0){15}}
		\multiput(81,32)(4,0){5}{\line(1,0){2}}
		\put(101,32){\line(1,0){15}}
		\put(119,32){\circle{6}}\put(113,17){\small$2n$}
		\put(121,30){\line(1,-1){20}}
		\put(143,8){\circle*{6}}\put(149,5){\small$2n+2$}
		\put(121,34){\line(1,1){20}}
		\put(143,56){\circle*{6}}\put(149,53){\small$2n+1$}
	\end{picture}
\caption{Gradation of $\mathfrak{g}(D^{(1)}_{2n+2})$ of type
$(1,1,0,1,0,\ldots,1,0,1,1)$}
\end{center}
\end{figure}
}

\title{Higher order Painlev\'{e} system of type $D^{(1)}_{2n+2}$ arising
from integrable hierarchy}
\author{Kenta Fuji and Takao Suzuki\\
{\small Department of Mathematics, Kobe University}\\
{\small Rokko, Kobe 657-8501, Japan}}
\date{}

\begin{document}

\maketitle

\begin{abstract}
A higher order Painlev\'{e} system of type $D^{(1)}_{2n+2}$ was introduced
by Y. Sasano.
It is an extension of the sixth Painlev\'{e} equation ($P_{\rm{VI}}$) for
the affine Weyl group symmetry.
It is also expressed as a Hamiltonian system of order $2n$ with a coupled
Hamiltonian of $P_{\rm{VI}}$.
In this paper, we discuss a derivation of this system from a
Drinfeld-Sokolov hierarchy.
\end{abstract}

\section{Introduction}

The Drinfeld-Sokolov hierarchies are extensions of the KdV (or mKdV)
hierarchy for the affine Lie algebras \cite{DS}.
It is known that they imply several Painlev\'{e} equations by similarity
reduction \cite{AS,FS,KK1,KIK,KK2}.
On the other hand, two types of extensions of the Painlev\'{e} equations
for the affine Weyl group symmetry have been studied, type $A^{(1)}_n$
\cite{NY1} and type $D^{(1)}_{2n+2}$ \cite{S}.
For type $A^{(1)}_n$ among them, the relation to the Drinfeld-Sokolov
hierarchies is already clarified.
In this paper, we investigate the relation for type $D^{(1)}_{2n+2}$.

Recall that the higher order Painlev\'{e} system of type $D^{(1)}_{2n+2}$
given in \cite{S} is a Hamiltonian system of order $2n$ with a coupled
Hamiltonian of $P_{\rm{VI}}$.
Let $q_i$, $p_i$ $(i=1,\ldots,n)$ be dependent variables on $s$ and
$\alpha_i$ $(i=0,\ldots,2n+2)$ complex parameters satisfying
\[
	\alpha_0 + \alpha_1 + \sum_{j=2}^{2n}2\alpha_j + \alpha_{2n+1}
	+ \alpha_{2n+2} = 1.
\]
We also set
\[\begin{split}
	H_i &= q_i(q_i-1)(q_i-s)p_i^2 - \{(\beta_{i,1}-1)q_i(q_i-1)\\
	&\quad +\beta_{i,3}(q_i-1)(q_i-s)+\beta_{i,4}q_i(q_i-s)\}p_i
	+ \alpha_{2i}(\alpha_{2i}+\beta_{i,0})q_i,
\end{split}\]
for $i=1,\ldots,n$, where
\[\begin{split}
	&\beta_{i,0} = \alpha_1 + \sum_{j=1}^{i-1}\alpha_{2j+1},\quad
	\beta_{i,1} = \alpha_0 + \sum_{j=1}^{i-1}2\alpha_{2j}
	+ \sum_{j=1}^{i-1}\alpha_{2j+1},\\
	&\beta_{i,3} = \sum_{j=i}^{n-1}\alpha_{2j+1}
	+ \sum_{j=i+1}^{n}2\alpha_{2j} + \alpha_{2n+1},\quad
	\beta_{i,4} = \sum_{j=i}^{n-1}\alpha_{2j+1} + \alpha_{2n+2}.
\end{split}\]
We consider a Hamiltonian system
\begin{equation}\label{Sys:Sasano}
	s(s-1)\frac{dq_i}{ds}=\{H,q_i\},\quad
	s(s-1)\frac{dp_i}{ds}=\{H,p_i\}\quad (i=1,\ldots,n),
\end{equation}
with a Hamiltonian
\begin{equation}\label{Ham:Sasano}
	H = \sum_{i=1}^{n}H_{i}
	+ \sum_{1\leq i<j\leq n}2(q_i-s)p_iq_j\{(q_j-1)p_j+\alpha_{2j}\},
\end{equation}
where $\{\cdot,\cdot\}$ stands for the Poisson bracket defined by
\[
	\{p_i,q_j\} = \delta_{i,j},\quad \{p_i,p_j\} = \{q_i,q_j\} = 0\quad
	(i,j=1,\ldots,n).
\]
Note that each $H_i$ is equivalent to the Hamiltonian of $P_{\rm{VI}}$
(see \cite{IKSY}).
In fact, the parameters satisfy the following relations:
\[
	\beta_{i,0}+\beta_{i,1}+2\alpha_{2i}+\beta_{i,3}+\beta_{i,4}=1\quad
	(i=1,\ldots,n).
\]

The system \eqref{Sys:Sasano} with \eqref{Ham:Sasano} admits affine Weyl
group symmetry of type $D^{(1)}_{2n+2}$.
Denoting the dependent variables by
\[\begin{split}
	&\varphi_0 = \frac{1}{2n+2},\quad \varphi_1 = q_1-s,\quad
	\varphi_{2i+1} = q_{i+1}-q_i\quad (i=1,\ldots,n-1),\\
	&\varphi_{2j} = -\frac{p_j}{2n+2}\quad (j=1,\ldots,n),\quad
	\varphi_{2n+1} = 1-q_n,\quad \varphi_{2n+2} = -q_n,
\end{split}\]
we consider birational canonical transformations
\begin{equation}\label{Aff_Wey_Sym}
	r_i(\alpha_j) = \alpha_j - a_{ij}\alpha_i,\quad
	r_i(\varphi_j) = \varphi_j + \frac{\alpha_i}{\varphi_i}
	\{\varphi_i,\varphi_j\},
\end{equation}
for $i,j=0,\ldots,2n+2$, where
\[\begin{array}{llll}
	a_{ii}=2& (i=0,\ldots,2n+2),\\[4pt]
	a_{02}=a_{ii+1}=a_{2n2n+2}=-1& (i=1,\ldots,2n),\\[4pt]
	a_{ij}=0& (\text{otherwise}).
\end{array}\]
Then the system \eqref{Sys:Sasano} with \eqref{Ham:Sasano} is invariant
under the action of them.
Furthermore, a group of symmetries $\langle r_0,\ldots,r_{2n+2}\rangle$
is isomorphic to the affine Weyl group of type $D^{(1)}_{2n+2}$.

In this paper, we show that the system \eqref{Sys:Sasano} with
\eqref{Ham:Sasano} is derived from a Drinfeld-Sokolov hierarchy by
similarity reduction.
The Drinfeld-Sokolov hierarchies are characterized by graded Heisenberg
subalgebras of the affine Lie algebras.
For a derivation of \eqref{Sys:Sasano}, we choose the affine Lie algebra 
$\mathfrak{g}(D_{2n+2}^{(1)})$ and its graded Heisenberg subalgebra of
type $(1,1,0,1,0,\ldots,1,0,1,1)$.
It is suggested by the fact that $P_{\rm{VI}}$ is derived from the hierarchy
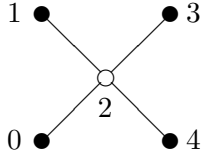
associated with the graded Heisenberg subalgebra of
$\mathfrak{g}(D_4^{(1)})$ of type $(1,1,0,1,1)$.\GradPVI\GradSas

This paper is organaized as follows.
In Section \ref{Sec:Lie_alg}, we recall the affine Lie algebra
$\mathfrak{g}(D^{(1)}_{2n+2})$ and its graded Heisenberg subalgebra.
In Section \ref{Sec:DS}, we formulate a similarity reduction of a
Drinfeld-Sokolov hierarchy of type $D^{(1)}_{2n+2}$.
In Section \ref{Sec:CP6}, we derive the system \eqref{Sys:Sasano} with
\eqref{Ham:Sasano} from the similarity reduction.
In Section \ref{Sec:Aff_Wey}, we discuss a derivation of the group of
symmetries \eqref{Aff_Wey_Sym}.

\section{Affine Lie algebra}\label{Sec:Lie_alg}

In this section, we introduce the affine Lie algebra of type
$D^{(1)}_{2n+2}$ and its Heisenberg subalgebra of type
$(1,1,0,1,0,\ldots,1,0,1,1)$, following the notation of \cite{Kac}.

Recall that $\mathfrak{g}=\mathfrak{g}(D^{(1)}_{2n+2})$ is a Lie algebra
generated by the Chevalley generators $e_i$, $f_i$, $\alpha_i^{\vee}$
$(i=0,\ldots,2n+2)$ and the scaling element $d$ with the fundamental
relations
\[\begin{split}
	&(\mathrm{ad}e_i)^{1-a_{ij}}(e_j)=0,\quad
	(\mathrm{ad}f_i)^{1-a_{ij}}(f_j)=0\quad (i\neq j),\\
	&[\alpha_i^{\vee},\alpha_j^{\vee}]=0,\quad
	[\alpha_i^{\vee},e_j]=a_{ij}e_j,\quad
	[\alpha_i^{\vee},f_j]=-a_{ij}f_j,\quad
	[e_i,f_j]=\delta_{i,j}\alpha_i^{\vee},\\
	&[d,\alpha_i^{\vee}]=0,\quad [d,e_i]=\delta_{i,0}e_0,\quad
	[d,f_i]=-\delta_{i,0}f_0,
\end{split}\]
for $i,j=0,\ldots,2n+2$.
The generalized Cartan matrix $A=\left(a_{ij}\right)_{i,j=0}^{2n+2}$
for $\mathfrak{g}$ is defined by
\[\begin{array}{llll}
	a_{ii}=2& (i=0,\ldots,2n+2),\\[4pt]
	a_{02}=a_{ii+1}=a_{2n2n+2}=-1& (i=1,\ldots,2n),\\[4pt]
	a_{ij}=0& (\text{otherwise}).
\end{array}\]
We denote the Cartan subalgebra of $\mathfrak{g}$ by
\[
	\mathfrak{h} = \bigoplus_{j=0}^{2n+2}\mathbb{C}\alpha_j^{\vee}
	\oplus\mathbb{C}d.
\]
The canonical central element of $\mathfrak{g}$ is given by
\[
	K = \alpha_0^{\vee} + \alpha_1^{\vee} + \sum_{i=2}^{2n}2\alpha_i^{\vee}
	+ \alpha_{2n+1}^{\vee} + \alpha_{2n+2}^{\vee}.
\]
The normalized invariant form
$(\cdot|\cdot):\mathfrak{g}\times\mathfrak{g}\to\mathbb{C}$ is determined
by the conditions
\[\begin{array}{lll}
	(\alpha_i^{\vee}|\alpha_j^{\vee}) = a_{ij},& (e_i|f_j) = \delta_{i,j},&
	(\alpha_i^{\vee}|e_j) = (\alpha_i^{\vee}|f_j) = 0,\\[4pt]
	(d|d) = 0,& (d|\alpha_j^{\vee}) = \delta_{0,j},& (d|e_j) = (d|f_j) = 0,
\end{array}\]
for $i,j=0,\ldots,2n+2$.

Consider a gradation $\mathfrak{g}=\bigoplus_{k\in\mathbb{Z}}\mathfrak{g}_k$
of type $(1,1,0,1,0,\ldots,1,0,1,1)$ by setting
\[\begin{array}{ll}
	\deg\mathfrak{h}=\deg e_i=\deg f_i=0& (i\in\mathcal{I}),\\[4pt]
	\deg e_j=1,\quad \deg f_j=-1& (j\in\mathcal{J}),
\end{array}\]
where $\mathcal{I}=\{2,4,\ldots,2n\}$ and
$\mathcal{J}=\{0,1,3,5,\ldots,2n+1,2n+2\}$.
With an element $\vartheta\in\mathfrak{h}$ such that
\[
	(\vartheta|\alpha_i^{\vee}) = 0,\quad
	(\vartheta|\alpha_j^{\vee}) = 1\quad (i\in\mathcal{I}; j\in\mathcal{J}),
\]
this gradation is defined by
\[
	\mathfrak{g}_k
	= \left\{x\in\mathfrak{g}\bigm|[\vartheta,x]=kx\right\}\quad
	(k\in\mathbb{Z}).
\]
We denote by
\[
	\mathfrak{g}_{<0} = \bigoplus_{k<0}\mathfrak{g}_{k},\quad
	\mathfrak{g}_{\geq0} = \bigoplus_{k\geq0}\mathfrak{g}_{k}.
\]

Such gradation implies the Heisenberg subalgebra of $\mathfrak{g}$
\[
	\mathfrak{s}
	= \{x\in\mathfrak{g}\bigm|[x,\Lambda_1]=[x,\Lambda_2]=\mathbb{C}K\},
\]
with elements of $\mathfrak{g}_1$
\[\begin{split}
	\Lambda_1 &= e_0 + e_{1,2} + \sum_{j\in\mathcal{J}'}(e_j+e_{j-1,j,j+1})
	+ e_{2n+1} + e_{2n,2n+2},\\
	\Lambda_2 &= e_1 + e_{0,2} + \sum_{j\in\mathcal{J}'}(e_{j-1,j}+e_{j,j+1})
	+ e_{2n+2} + e_{2n,2n+1},
\end{split}\]
where $\mathcal{J}'=\{3,5,\ldots,2n-1\}$ and
\[
	e_{i_1,i_2,\ldots,i_{n-1},i_n} = \mathrm{ad}e_{i_1}\mathrm{ad}e_{i_2}
	\ldots\mathrm{ad}e_{i_{n-1}}(e_{i_n}).
\]
Note that $\mathfrak{s}$ admits the gradation of type
$(1,1,0,1,0,\ldots,1,0,1,1)$, namely
\[
	\mathfrak{s} = \bigoplus_{k\in\mathbb{Z}}\mathfrak{s}_k,\quad
	\mathfrak{s}_k\subset\mathfrak{g}_k.
\]
We also remark that the positive part of $\mathfrak{s}$ has a graded bases
$\left\{\Lambda_k\right\}_{k=1}^{\infty}$ satisfying
\[
	[\Lambda_k,\Lambda_l] = 0,\quad
	[\vartheta,\Lambda_k] = n_k\Lambda_k\quad (k,l=1,2,\ldots),
\]
where $n_k$ stands for the degree of element $\Lambda_k$ defined by
\[
	n_k = \left\{\begin{array}{ll}
		k& (k:\text{odd})\\[4pt]
		k-1	& (k:\text{even})
	\end{array}\right..
\]
The explicit formulas of $\Lambda_k$ $(k\geq3)$ are given in Appendix
\ref{Sec:Hei_subalg}.

In the last, we introduce the Borel subalgebra of $\mathfrak{g}$.
Let $\mathfrak{n}_{+}$ and $\mathfrak{n}_{-}$ be the subalgebras of
$\mathfrak{g}$ generated by $e_i$ and $f_i$ $(i=0,\ldots,2n+2)$
respectively.
Then the Borel subalgebra $\mathfrak{b}_{+}$ of $\mathfrak{g}$ is defined
by $\mathfrak{b}_{+}=\mathfrak{h}\oplus\mathfrak{n}_{+}$.
Note that we have the triangular decomposition
\[
	\mathfrak{g} = \mathfrak{n}_{-}\oplus\mathfrak{h}\oplus\mathfrak{n}_{+}
	= \mathfrak{n}_{-}\oplus\mathfrak{b}_{+}.
\]
We also remark that
\[
	\mathfrak{n}_{-} = \mathfrak{g}_{<0}
	\oplus\bigoplus_{i\in\mathcal{I}}\mathbb{C}f_i,\quad
	\mathfrak{g}_{\geq0} = \bigoplus_{i\in\mathcal{I}}\mathbb{C}f_i
	\oplus\mathfrak{b}_{+}.
\]

\section{Drinfeld-Sokolov hierarchy}\label{Sec:DS}

In this section, we formulate a Drinfeld-Sokolov hierarchy of type
$D^{(1)}_{2n+2}$ and its similarity reduction associated with the
Heisenberg subalgebra $\mathfrak{s}$.

In the following, we use the notation of infinite dimensional groups
\[
	G_{<0} = \exp(\wh{\mathfrak{g}}_{<0}),\quad
	G_{\geq0} = \exp(\wh{\mathfrak{g}}_{\geq0}),
\]
where $\wh{\mathfrak{g}}_{<0}$ and $\wh{\mathfrak{g}}_{\geq0}$ are
completions of $\mathfrak{g}_{<0}$ and $\mathfrak{g}_{\geq0}$ respectively.

Introducing the time variables $t_k$ $(k=1,2,\ldots)$, we consider a system
of partial differential equations
\begin{equation}\label{Eq:Sato}
	\partial_{t_k} - B_k = W(\partial_{t_k}-\Lambda_k)W^{-1}\quad
	(k=1,2,\ldots),
\end{equation}
for a $G_{<0}$-valued function $W$, where $B_k$ stands for the
$\mathfrak{g}_{\geq0}$-component of $W\Lambda_kW^{-1}$.
The Zakharov-Shabat equations
\begin{equation}\label{ZS_DS}
	[\partial_{t_k}-B_k,\partial_{t_l}-B_l] = 0\quad (k,l=1,2,\ldots),
\end{equation}
follows from the system \eqref{Eq:Sato}.
We call the system \eqref{ZS_DS} the Drinfeld-Sokolov hierarchy of type
$D^{(1)}_{2n+2}$.

Under the system \eqref{Eq:Sato}, we consider the operator
\[
	\mathcal{M} = W
	\exp\left(\sum_{k=1,2,\ldots}t_k\Lambda_k\right)\vartheta
	\exp\left(-\sum_{k=1,2,\ldots}t_k\Lambda_k\right)W^{-1}.
\]
Then the operator $\mathcal{M}$ satisfies
\begin{equation}\label{Sim_Red}
	[\partial_{t_k}-B_k,\mathcal{M}] = 0\quad (k=1,2,\ldots).
\end{equation}
Also $\mathcal{M}$ is expressed as
\[
	\mathcal{M} = W\vartheta W^{-1}
	- \sum_{k=1,2,\ldots}n_kt_kW\Lambda_kW^{-1}.
\]

Now we require that the similarity condition
$\mathcal{M}\in\mathfrak{g}_{\geq0}$ is satisfied.
Note that it is equivalent to
\[
	\vartheta + \sum_{k=1,2,\ldots}n_kt_kB^c_k = W\vartheta W^{-1},
\]
where $B^c_k$ stands for the $\mathfrak{g}_{<0}$-component of
$W\Lambda_kW^{-1}$.
Then we have
\[
	\mathcal{M} = \vartheta - \sum_{k=1,2,\ldots}n_kt_kB_k.
\]
We also assume that $t_k=0$ for $k\geq3$.
Then the systems \eqref{ZS_DS} and \eqref{Sim_Red} are equivalent to
\begin{equation}\begin{split}\label{Sim_Red_2}
	&[\partial_{t_1}-B_1,\partial_{t_2}-B_2] = 0,\\
	&[\partial_{t_k}-B_k,\vartheta-t_1B_1-t_2B_2] = 0\quad (k=1,2).
\end{split}\end{equation}
We regard the system \eqref{Sim_Red_2} as a similarity reduction of the
Drinfeld-Sokolov hierarchy of type $D^{(1)}_{2n+2}$.

The $\mathfrak{g}_{\geq0}$-valued functions $B_k$ $(k=1,2)$ are expressed in
the form
\[
	B_k = U_k + \Lambda_k,\quad
	U_k = \sum_{i=0}^{2n+2}u_{k,i}\alpha_i^{\vee}
	+ \sum_{i\in\mathcal{I}}x_{k,i}e_i + \sum_{i\in\mathcal{I}}y_{k,i}f_i.
\]
In terms of the operators $U_k\in\mathfrak{g}_0$, this similarity reduction
can be expressed as
\[\begin{split}
	&\partial_{t_1}(U_2) - \partial_{t_2}(U_1) + [U_2,U_1] = 0,\\
	&[\Lambda_1,U_2] - [\Lambda_2,U_1] = 0,\\
	&t_1\partial_{t_1}(U_k) + t_2\partial_{t_2}(U_k) + U_k = 0\quad (k=1,2).
\end{split}\]

In the following, we use the notation of a $\mathfrak{g}_{\geq0}$-valued
1-form $\mathcal{B}=B_1dt_1+B_2dt_2$ with respect to the coordinates
$\bm{t}=(t_1,t_2)$.
Then the similarity reduction \eqref{Sim_Red_2} is expressed as
\begin{equation}\label{Sim_Red_ED}
	d_{\bm{t}}\mathcal{M} = [\mathcal{B},\mathcal{M}],\quad
	d_{\bm{t}}\mathcal{B} = \mathcal{B}\wedge\mathcal{B},
\end{equation}
where $d_{\bm{t}}$ stands for an exterior differentation with respect to
$\bm{t}$.
Denoting by
\[
	\mathcal{M}_1 = -t_1\Lambda_1 - t_2\Lambda_2,\quad
	\mathcal{B}_1 = \Lambda_1dt_1 + \Lambda_2dt_2,
\]
we can express the operators $\mathcal{M}$ and $\mathcal{B}$ in the form
\[\begin{split}
	\mathcal{M} &= \theta + \sum_{i\in\mathcal{I}}\xi_ie_i
	+ \sum_{i\in\mathcal{I}}\psi_if_i + \mathcal{M}_1,\\
	\mathcal{B} &= \bm{u} + \sum_{i\in\mathcal{I}}\bm{x}_ie_i
	+ \sum_{i\in\mathcal{I}}\bm{y}_if_i + \mathcal{B}_1,
\end{split}\]
where
\[
	\theta = \vartheta + \sum_{i=0}^{2n+2}\theta_i\alpha^{\vee}_i,\quad
	\bm{u} = \sum_{i=0}^{2n+2}\bm{u}_i\alpha^{\vee}_i.
\]
The system \eqref{Sim_Red_ED} is expressed in terms of these variables as
follows:
\[\begin{split}
	&d_{\bm{t}}\theta_i = \bm{x}_i\psi_i - \bm{y}_i\xi_i,\quad
	d_{\bm{t}}\theta_j = 0,\\
	&d_{\bm{t}}\xi_i = (\bm{u}|\alpha^{\vee}_i)\xi_i
	- \bm{x}_i(\theta|\alpha^{\vee}_i),\\
	&d_{\bm{t}}\psi_i = -(\bm{u}|\alpha^{\vee}_i)\psi_i
	+ \bm{y}_i(\theta|\alpha^{\vee}_i),
\end{split}\]
and
\[\begin{split}
	&d_{\bm{t}}\bm{u}_i = \bm{x}_i\wedge\bm{y}_i+\bm{y}_i\wedge\bm{x}_i,\quad
	d_{\bm{t}}\bm{u}_j = 0,\\
	&d_{\bm{t}}\bm{x}_i = (\bm{u}|\alpha^{\vee}_i)\wedge\bm{x}_i,\quad
	d_{\bm{t}}\bm{y}_i = -(\bm{u}|\alpha^{\vee}_i)\wedge\bm{y}_i,
\end{split}\]
for $i\in\mathcal{I}$ and $j\in\mathcal{J}$.

\section{Coupled Painlev\'{e} VI system}\label{Sec:CP6}

In this section, we show that the system \eqref{Sys:Sasano} with
\eqref{Ham:Sasano} is derived from the similarity reduction
\eqref{Sim_Red_ED}.

We introduce below a gauge transformation
\[
	\mathcal{M}^{+} = \exp(\mathrm{ad}(\Gamma))\mathcal{M},\quad
	d_{\bm{t}}-\mathcal{B}^{+}
	= \exp(\mathrm{ad}(\Gamma))(d_{\bm{t}}-\mathcal{B}),
\]
with $\Gamma\in\mathfrak{g}_0$ such that $\mathcal{M}^{+}$ and
$\mathcal{B}^{+}$ should take values in $\mathfrak{b}_{+}$.
Then the system \eqref{Sim_Red_ED} is transformed into
\[
	d_{\bm{t}}\mathcal{M}^{+} = [\mathcal{B}^{+},\mathcal{M}^{+}],\quad
	d_{\bm{t}}\mathcal{B}^{+} = \mathcal{B}^{+}\wedge\mathcal{B}^{+}.
\]
It is equivalent to the system \eqref{Sys:Sasano} with \eqref{Ham:Sasano}
under a certain transformation of variables.
We recall that the operator $\mathcal{M}$ is expressed as
\[
	\mathcal{M} = \theta + \sum_{i\in\mathcal{I}}\xi_ie_i
	+ \sum_{i\in\mathcal{I}}\psi_if_i + \mathcal{M}_1,
\]
where
\[\begin{split}
	\mathcal{M}_1 &= -t_1e_0 - t_2e_1 - \sum_{j\in\mathcal{J}'}t_1e_j
	- t_1e_{2n+1} - t_2e_{2n+2} - t_2e_{0,2} - t_1e_{1,2}\\
	&\quad - \sum_{j\in\mathcal{J}'}t_2(e_{j-1,j}+e_{j,j+1}) - t_2e_{2n,2n+1}
	- t_1e_{2n,2n+2} - \sum_{j\in\mathcal{J}'}t_1e_{j-1,j,j+1}.
\end{split}\]

We first consider a gauge transformation
\[
	\mathcal{M}' = \exp(\mathrm{ad}(\Gamma_1))\mathcal{M},\quad
	d_{\bm{t}}-\mathcal{B}'
	= \exp(\mathrm{ad}(\Gamma_1))(d_{\bm{t}}-\mathcal{B}),
\]
with $\Gamma_1=\sum_{i\in\mathcal{I}}\gamma_ie_i$ defined by
\[
	\gamma_2 = \frac{t_2}{t_1},\quad
	\gamma_{2i+2} = \frac{t_1+t_2\gamma_{2i}}{t_2+t_1\gamma_{2i}}\quad
	(i=1,\ldots,n-1).
\]
Then we obtain
\[\begin{split}
	\mathcal{M}_1' &= \exp(\mathrm{ad}(\Gamma_1))(\mathcal{M}_1)\\
	&= -t_1e_0 - t_2e_1 - \sum_{j\in\mathcal{J}'}t_1e_j - t_2e_{2n+1}
	- t_1e_{2n+2}\\
	&\quad - (t_1-t_2\gamma_2)e_{1,2}
	- \sum_{j\in\mathcal{J}'}\{(t_2+t_1\gamma_{j-1})e_{j-1,j}
	+(t_2-t_1\gamma_{j+1})e_{j,j+1}\}\\
	&\quad - (t_1+t_2\gamma_{2n})e_{2n,2n+1}
	- (t_2+t_1\gamma_{2n})e_{2n,2n+2}.
\end{split}\]

We next consider a gauge transformation
\[
	\mathcal{M}^* = \exp(\mathrm{ad}(\Gamma_2))\mathcal{M}',\quad
	d_{\bm{t}}-\mathcal{B}^*
	= \exp(\mathrm{ad}(\Gamma_2))(d_{\bm{t}}-\mathcal{B}'),
\]
with $\Gamma_2\in\mathfrak{h}$ such that
\[\begin{split}
	\mathcal{M}_1^* &= \exp(\mathrm{ad}(\Gamma_2))(\mathcal{M}_1')\\
	&= e_0 + b_1e_1 + \sum_{j\in\mathcal{J}'}b_je_j + b_{2n+1}e_{2n+1}
	+ b_{2n+2}e_{2n+2}\\
	&\quad + e_{1,2} + \sum_{j\in\mathcal{J}'}(e_{j-1,j}+e_{j,j+1})
	+ e_{2n,2n+1} + e_{2n,2n+2}.
\end{split}\]
Note that the coefficients $b_j$ are algebraic functions in $t_1$ and $t_2$.
Then We have
\begin{equation}\label{Sim_Red_ED_2}
	d_{\bm{t}}\mathcal{M}^* = [\mathcal{B}^*,\mathcal{M}^*],\quad
	d_{\bm{t}}\mathcal{B}^* = \mathcal{B}^*\wedge\mathcal{B}^*.
\end{equation}
With the notation
\[
	\mathcal{B}_1^* = \exp(\mathrm{ad}(\Gamma_2))
	\exp(\mathrm{ad}(\Gamma_1))(\mathcal{B}_1),
\]
the operators $\mathcal{M}^*$ and $\mathcal{B}^*$ are expressed in the form
\[\begin{split}
	\mathcal{M}^* &= \theta^* + \sum_{i\in\mathcal{I}}\xi^*_ie_i
	+ \sum_{i\in\mathcal{I}}\psi^*_if_i + \mathcal{M}_1^*,\\
	\mathcal{B}^* &= \bm{u}^* + \sum_{i\in\mathcal{I}}\bm{x}^*_ie_i
	+ \sum_{i\in\mathcal{I}}\bm{y}^*_if_i + \mathcal{B}_1^*,
\end{split}\]
where
\[
	\theta^* = \vartheta + \sum_{i=0}^{2n+2}\theta^*_i\alpha^{\vee}_i,\quad
	\bm{u}^* = \sum_{i=0}^{2n+2}\bm{u}^*_i\alpha^{\vee}_i.
\]

We finally consider a gauge transformation
\[
	\mathcal{M}^{+} = \exp(\mathrm{ad}(\Gamma_3))\mathcal{M}^*,\quad
	d_{\bm{t}}-\mathcal{B}^{+}
	= \exp(\mathrm{ad}(\Gamma_3))(d_{\bm{t}}-\mathcal{B}^*),
\]
with $\Gamma_3=\sum_{i\in\mathcal{I}}\eta_if_i$ such that
$\mathcal{M}^{+},\mathcal{B}^{+}\in\mathfrak{b}_{+}$, namely
\begin{equation}\label{Gauge_Trf_1}
	\xi^*_i\eta_i^2 - (\theta^*|\alpha^{\vee}_i)\eta_i - \psi^*_i = 0\quad
	(i\in\mathcal{I}),
\end{equation}
and
\begin{equation}\label{Gauge_Trf_2}
	d_{\bm{t}}\eta_i = \bm{x}^*_i\eta_i^2 - (\bm{u}^*|\alpha^{\vee}_i)\eta_i
	- \bm{y}^*_i\quad (i\in\mathcal{I}).
\end{equation}
Here we have
\begin{lem}\label{Lem:Borel}
Under the system \eqref{Sim_Red_ED_2}, the equation \eqref{Gauge_Trf_2}
follows from the equation \eqref{Gauge_Trf_1}.
\end{lem}
\begin{proof}
The first equation of the system \eqref{Sim_Red_ED_2} can be expressed as
\begin{equation}\begin{split}\label{Sim_Red_ED_3}
	&d_{\bm{t}}\theta^*_i = \bm{x}^*_i\psi^*_i - \bm{y}^*_i\xi^*_i,\quad
	d_{\bm{t}}\theta^*_j = 0,\\
	&d_{\bm{t}}\xi^*_i = (\bm{u}^*|\alpha^{\vee}_i)\xi^*_i
	- \bm{x}^*_i(\theta^*|\alpha^{\vee}_i),\\
	&d_{\bm{t}}\psi^*_i = -(\bm{u}^*|\alpha^{\vee}_i)\psi^*_i
	+ \bm{y}^*_i(\theta^*|\alpha^{\vee}_i),
\end{split}\end{equation}
for $i\in\mathcal{I}$ and $j\in\mathcal{J}$.
By using \eqref{Sim_Red_ED_3} and
$(d_{\bm{t}}\theta^*|\alpha^{\vee}_i)=2d_{\bm{t}}\theta^*_i$, we obtain
\[\begin{split}
	&d_{\bm{t}}\{\xi^*_i\eta_i^2-(\theta^*|\alpha^{\vee}_i)\eta_i
	-\psi^*_i\}\\
	&= \{2\xi^*_i\eta_i-(\theta^*|\alpha^{\vee}_i)\}
	\{d_{\bm{t}}\eta_i-\bm{x}^*_i\eta_i^2+(\bm{u}^*|\alpha^{\vee}_i)\eta_i
	+\bm{y}^*_i\}\quad (i\in\mathcal{I}).
\end{split}\]
It follows that the equation \eqref{Gauge_Trf_1} implies
\eqref{Gauge_Trf_2} or
\begin{equation}\label{Prf:Lem:Borel_1}
	\eta_i = \frac{(\theta^*|\alpha^{\vee}_i)}{2\xi^*_i}\quad
	(i\in\mathcal{I}).
\end{equation}
Hence it is enough to verify that the equation \eqref{Gauge_Trf_2}
follows from \eqref{Prf:Lem:Borel_1}.
Together with \eqref{Sim_Red_ED_3}, the equation \eqref{Prf:Lem:Borel_1}
implies
\begin{equation}\begin{split}\label{Prf:Lem:Borel_2}
	d_{\bm{t}}\eta_i &= \frac{(d_{\bm{t}}\theta^*|\alpha^{\vee}_i)\xi^*_i
	-(\theta^*|\alpha^{\vee}_i)d_{\bm{t}}\xi^*_i}{2(\xi^*_i)^2}\\
	&= \bm{x}^*_i\eta_i^2 - (\bm{u}^*|\alpha^{\vee}_i)\eta_i - \bm{y}^*_i
	+ \frac{\bm{x}^*_i\{4\xi^*_i\psi^*_i+(\theta^*|\alpha^{\vee}_i)^2\}}
	{4(\xi^*_i)^2}\quad (i\in\mathcal{I}).
\end{split}\end{equation}
On the other hand, we obtain
\begin{equation}\label{Prf:Lem:Borel_3}
	4\xi^*_i\psi^*_i + (\theta^*|\alpha^{\vee}_i)^2 = 0\quad
	(i\in\mathcal{I}),
\end{equation}
by substituting \eqref{Prf:Lem:Borel_1} into \eqref{Gauge_Trf_1}.
Combining \eqref{Prf:Lem:Borel_2} and \eqref{Prf:Lem:Borel_3}, we obtain
the equation \eqref{Gauge_Trf_2}.
\end{proof}

Thanks to Lemma \ref{Lem:Borel}, the gauge parameters $\eta_i$
$(i\in\mathcal{I})$ are determined by the equation \eqref{Gauge_Trf_1}.
Hence we obtain the system on $\mathfrak{b}_{+}$
\begin{equation}\label{Sim_Red_Borel}
	d_{\bm{t}}\mathcal{M}^{+} = [\mathcal{B}^{+},\mathcal{M}^{+}],\quad
	d_{\bm{t}}\mathcal{B}^{+} = \mathcal{B}^{+}\wedge\mathcal{B}^{+},
\end{equation}
with dependent variables
\[
	\lambda_i = \eta_i - \sum_{j=1}^{i-1}b_{2j+1},\quad
	\mu_i = \varphi^*_i\quad (i\in\mathcal{I}).
\]
The operator $\mathcal{M}^{+}$ is expressed in the form
\[\begin{split}
	\mathcal{M}^{+} &= \kappa + \sum_{i\in\mathcal{I}}\mu_ie_i + e_0
	+ (c_1-\lambda_2)e_1
	+ \sum_{j\in\mathcal{J}'}(\lambda_{j-1}-\lambda_{j+1})e_j\\
	&\quad + (\lambda_{2n}-c_{2n+1})e_{2n+1}
	+ (\lambda_{2n}-c_{2n+2})e_{2n+2}\\
	&\quad + e_{1,2} + \sum_{j\in\mathcal{J}'}(e_{j-1,j}+e_{j,j+1})
	+ e_{2n,2n+1} + e_{2n,2n+2},
\end{split}\]
where $\kappa\in\mathfrak{h}$ and
\[
	c_1 = b_1,\quad
	c_i = -\sum_{j=1}^{n-1}b_{2j+1} - b_i\quad (i=2n+1,2n+2).
\]
Note that $d_{\bm{t}}\kappa=0$.
We also remark that $c_1$, $c_{2n+1}$ and $c_{2n+2}$ are algebraic functions
in $t_1$ and $t_2$.

Let
\[
	s_1 = \frac{c_{2n+2}-c_1}{2n+2},\quad
	s_2 = \frac{c_{2n+2}-c_{2n+1}}{2n+2}.
\]
We now regard the system \eqref{Sim_Red_Borel} as a system of ordinary
differential equations
\begin{equation}\label{Sim_Red_ODE}
	\left[s(s-1)\frac{d}{ds}-B,\mathcal{M}^{+}\right] = 0,
\end{equation}
with respect to the independent variable $s=s_1$ by setting $s_2=1$.
The explicit formula of the $\mathfrak{b}_{+}$-valued operator $B$ is
given below.
We also set
\[
	q_i = \frac{c_{2n+2}-\lambda_{2i}}{2n+2},\quad p_i = -\mu_{2i},\quad
	\alpha_j = \frac{(\kappa|\alpha^{\vee}_j)}{2n+2},
\]
for $i=1,\ldots,n$ and $j=0,\ldots,2n+2$.
Then we obtain
\begin{thm}
The system \eqref{Sim_Red_ODE} is equivalent to the system
\eqref{Sys:Sasano} with \eqref{Ham:Sasano}.
\end{thm}
The operator $\mathcal{M}^{+}$ is described as
\[\begin{split}
	\mathcal{M}^{+} &= \kappa + \sum_{i=0}^{2n+2}(2n+2)\varphi_ie_i
	+ \sum_{i=1}^{2n}e_{i,i+1} + e_{2n,2n+2},
\end{split}\]
We recall that
\[\begin{split}
	&\varphi_0 = \frac{1}{2n+2},\quad \varphi_1 = q_1-s,\quad
	\varphi_{2i+1} = q_{i+1}-q_i\quad (i=1,\ldots,n-1),\\
	&\varphi_{2j} = -\frac{p_j}{2n+2}\quad (j=1,\ldots,n),\quad
	\varphi_{2n+1} = 1-q_n,\quad \varphi_{2n+2} = -q_n.
\end{split}\]
The operator $B$ is described as
\[\begin{split}
	B &= u + \sum_{i=0}^{2n+2}x_ie_i + y_1e_{0,2}
	+ \sum_{i=2}^{2n}y_ie_{i,i+1} + y_{2n+1}e_{2n,2n+2}
	+ \sum_{j\in\mathcal{J}'}y_1e_{j-1,j,j+1},
\end{split}\]
where
\[\begin{split}
	&x_0 = -\frac{q_1-s}{2n+2},\quad x_1 = 1,\quad
	x_{2i+1} = s(s-1)-(q_i-s)(q_{i+1}-s),\\
	&x_{2n+1} = -(s-1)q_n,\quad x_{2n+2} = -s(q_n-1),\\
	&y_1 = -\frac{1}{(2n+2)^2},\quad
	y_{2i} = -\frac{q_{i+1}-s}{2n+2},\quad
	y_{2i+1} = \frac{q_i-s}{2n+2},\\
	&y_{2n} = \frac{s-1}{2n+2},\quad y_{2n+1} = \frac{s}{2n+2},
\end{split}\]
for $i=1,\ldots,n-1$ and
\[\begin{split}
	(2n+2)x_{2i} &= \sum_{j=1}^{i-1}2\{(q_j-s)p_j+\alpha_{2j}\}
	+ (q_i-s)p_i + \alpha_{2i} + \alpha_0 + \sum_{j=1}^{i-1}\alpha_{2j+1},
\end{split}\]
for $i=1,\ldots,n$.
Here $u=\sum_{i=0}^{2n+2}u_i\alpha^{\vee}_i$ satisfies
\[\begin{split}
	(u|\alpha^{\vee}_0) &= -\alpha_0(q_1-s),\\
	(u|\alpha^{\vee}_1) &= - \alpha_0(q_1+s-1)
	- \sum_{j=1}^n2q_j\{(q_j-1)p_j+\alpha_{2j}\}\\
	&\quad - (2\alpha_2+\beta_{1,3})(s-1) - \beta_{1,4}s,\\
	(u|\alpha^{\vee}_{2i+1}) &= -\left\{
	\sum_{j=1}^i2(q_j-s)p_j+\beta_{i,1}+2\alpha_{2i}\right\}(q_i+q_{i+1}-1)
	- \beta_{i+1,4}s\\
	&\quad -\sum_{j=i+1}^n2q_j\{(q_j-1)p_j+\alpha_{2j}\}
	- (2\alpha_{2i+2}+\beta_{i+1,3})(s-1),\\
	(u|\alpha^{\vee}_{2n+1}) &= -\left\{\sum_{j=1}^n2(q_j-s)p_j
	+\beta_{n,1}+2\alpha_{2n}\right\}q_n - \alpha_{2n+1}s,\\
	(u|\alpha^{\vee}_{2n+2}) &= -\left\{\sum_{j=1}^n2(q_j-s)p_j
	+\beta_{n,1}+2\alpha_{2n}\right\}(q_n-1) - \alpha_{2n+2}(s-1),
\end{split}\]
for $i=1,\ldots,n-1$ and
\[\begin{split}
	(u|\alpha^{\vee}_{2i}) &= \left\{\sum_{j=1}^{i-1}2(q_j-s)p_j
	+(q_i-s)p_i+\beta_{i,1}+2\alpha_{2i}\right\}(2q_i-1)\\
	&\quad + q_i\{(q_i-1)p_i+\alpha_{2i}\}
	+ \sum_{j=i+1}^n2q_j\{(q_j-1)p_j+\alpha_{2j}\}\\
	&\quad + (2\alpha_{2i}+\beta_{i+1,3})(s-1) + \beta_{i+1,4}s,\\
\end{split}\]
for $i=1,\ldots,n$, where
\[\begin{split}
	&\beta_{i,0} = \alpha_1 + \sum_{j=1}^{i-1}\alpha_{2j+1},\quad
	\beta_{i,1} = \alpha_0 + \sum_{j=1}^{i-1}2\alpha_{2j}
	+ \sum_{j=1}^{i-1}\alpha_{2j+1},\\
	&\beta_{i,3} = \sum_{j=i}^{n-1}\alpha_{2j+1}
	+ \sum_{j=i+1}^{n}2\alpha_{2j} + \alpha_{2n+1},\quad
	\beta_{i,4} = \sum_{j=i}^{n-1}\alpha_{2j+1} + \alpha_{2n+2}.
\end{split}\]

\begin{rem}
The system \eqref{Sys:Sasano} with \eqref{Ham:Sasano} is derived from
a Lax pair associated with the loop algebra
$\mathfrak{so}(4n+4)[z,z^{-1}]${\rm;} see Appendix \ref{Sec:Lax}.
\end{rem}

\section{Affine Weyl group symmetry}\label{Sec:Aff_Wey}

In this section, we discuss a derivation of the group of symmetries
\eqref{Aff_Wey_Sym} following the manner in \cite{NY2}.

Recall that the affine Weyl group of type $D^{(1)}_{2n+2}$ is generated by
the transformations $r_i$ $(i=0,\ldots,2n+2)$ with the fundamental relations
\[\begin{array}{ll}
	r_i^2=1& (i=0,\ldots,2n+2),\\[4pt]
	(r_ir_j)^{2-a_{ij}}=0& (i,j=0,\ldots,2n+2; i\neq j).
\end{array}\]
acting on the simple roots as
\[
	r_i(\alpha_j) = \alpha_j - a_{ij}\alpha_i\quad (i,j=0,\ldots,2n+2),
\]
where
\[\begin{array}{llll}
	a_{ii}=2& (i=0,\ldots,2n+2),\\[4pt]
	a_{02}=a_{ii+1}=a_{2n2n+2}=-1& (i=1,\ldots,2n),\\[4pt]
	a_{ij}=0& (\text{otherwise}).
\end{array}\]

Let $X(0)\in G_{<0}G_{\geq0}$.
We consider a $G_{<0}G_{\geq0}$-valued function
\[
	X = X(t_1,t_2,\ldots)
	= \exp\left(\sum_{k=1,2,\ldots}t_k\Lambda_k\right)X(0).
\]
Then we have a system of partial differential equations
\[
	X\partial_kX^{-1} = \partial_k- \Lambda_k\quad (k=1,2,\ldots),
\]
defined through the adjoint action of $G_{<0}G_{\geq0}$ on
$\wh{\mathfrak{g}}_{<0}\oplus\mathfrak{g}_{\geq0}$.
Via a decomposition
\[
	X = W^{-1}Z,\quad W\in G_{<0},\quad Z\in G_{\geq0},
\]
we obtain the system \eqref{Eq:Sato}.

In the previous section, we have considered the gauge transformation
\[
	\mathcal{M}^{+} = \exp(\mathrm{ad}(\Gamma))\mathcal{M},\quad
	d_{\bm{t}}-\mathcal{B}^{+}
	= \exp(\mathrm{ad}(\Gamma))(d_{\bm{t}}-\mathcal{B}),\quad
	\Gamma\in\mathfrak{g}_0,
\]
for the derivation of the system \eqref{Sys:Sasano}.
Note that it arises from
\[
	X = (W^{+})^{-1}Z^{+},\quad W^{+} = \exp(\Gamma)W,\quad
	Z^{+} = \exp(\Gamma)Z.
\]

Consider transformations
\[
	r_i(X) = X\exp(-e_i)\exp(f_i)\exp(-e_i)\quad (i=0,\ldots,2n+2).
\]
Under the similarity condition $\mathcal{M}^{+}\in\mathfrak{b}_{+}$,
their action on $W^{+}$ is given by
\[
	r_i(W^{+}) = G_iW^{+}\quad (i=0,\ldots,2n+2),
\]
where
\[
	G_i = \exp\left(\frac{\alpha_i}{\varphi_i}f_i\right),\quad
	\alpha_i = \frac{(\alpha_i^{\vee}|\mathcal{M}^{+})}{2n+2},\quad
	\varphi_i = \frac{(f_i|\mathcal{M}^{+})}{2n+2}.
\]
It follows that
\[
	r_i(\mathcal{M}^{+}) = G_i\mathcal{M}^{+}G_i^{-1},\quad
	d_{\bm{t}}-r_i(\mathcal{B}^{+}) = G_i(d_{\bm{t}}-\mathcal{B}^{+})G_i^{-1},
\]
for $i=0,\ldots,2n+2$.
Then each $r_i(\mathcal{M}^{+})$ and $r_i(\mathcal{B}^{+})$ are
$\mathfrak{b}_{+}$-valued and satisfy the system \eqref{Sim_Red_Borel}.
Note that the complex parameters $\alpha_i$ $(i=0,\ldots,2n+2)$ can be
regarded as the simple roots for $\mathfrak{g}(D^{(1)}_{2n+2})$.

We define a Poisson structure for the operator $\mathcal{M}^{+}$ by
\[
	\{\varphi_i,\varphi_j\} = \frac{([f_j,f_i]|\mathcal{M}^{+})}{2n+2}\quad
	(i,j=0,\ldots,2n+2).
\]
It is equivalent to
\[
	\{p_i,q_j\} = \delta_{i,j},\quad \{p_i,p_j\} = \{q_i,q_j\} = 0\quad
	(i,j=1,\ldots,n).
\]
Hence $p_i$, $q_i$ $(i=1,\ldots,n)$ give a canonical coordinate system
associated with the Poisson structure for $\mathcal{M}^{+}$.
Then the action of the transformations $r_i$ $(i=0,\ldots,2n+2)$ on the
coefficients of $\mathcal{M}^{+}$ is equivalent to \eqref{Aff_Wey_Sym}.

\begin{rem}[\cite{S}]
Let
\[
	\sigma_1(i) = (0,1)(2n+1,2n+2)i,\quad \sigma_2(i) = 2n+2-i\quad
	(i=0,\ldots,2n+2),
\]
where $(0,1)$ and $(2n+1,2n+2)$ stand for the adjacent transpositions.
Then the system \eqref{Sys:Sasano} with \eqref{Ham:Sasano} is invariant
under the action of transformations $\pi_1$ and $\pi_2$ defined by
\[\begin{split}
	\pi_1(\alpha_i) = \alpha_{\sigma_1(i)},\quad
	\pi_1(q_i) = \frac{s(q_i-1)}{q_i-s},\quad
	\pi_1(p_i) = \frac{(q_i-s)\{p_i(q_i-s)+\alpha_{2i}\}}{s(1-s)},
\end{split}\]
and
\[
	\pi_2(\alpha_i) = \alpha_{\sigma_2(i)},\quad
	\pi_2(q_i) = \frac{s}{q_i},\quad
	\pi_2(p_i) = -\frac{q_i(q_ip_i+\alpha_{2i})}{s},
\]
for $i=0,\ldots,2n+2$.
These transformations generate a group of Dynkin diagram automorphisms
of type $D^{(1)}_{2n+2}$.
In fact, they satisfy the fundamental relations
\[\begin{array}{ll}
	\pi_i^2=1& (i=1,2),\\[4pt]
	(\pi_1\pi_2)^3=1,\\[4pt]
	\pi_ir_j=r_{\sigma_i(j)}\pi_i& (i=1,2; j=0,\ldots,2n+2).
\end{array}\]
\end{rem}

\appendix

\section{Heisenberg subalgebra}\label{Sec:Hei_subalg}

We first introduce the simple Lie algebra $\mathfrak{so}(4n+4)$ and
its loop algebra.
Denoting matrix units by
$E_{i,j}=\left(\delta_{i,k}\delta_{j,l}\right)_{k,l=1}^{4n+4}$, we set
\[
	J = \sum_{i=1}^{4n+4}E_{i,4n+5-i}.
\]
Then the algebra $\mathfrak{so}(4n+4)$ is defined by
\[
	\mathfrak{so}(4n+4) = \left\{X\in\mathrm{Mat}(4n+4;\mathbb{C})\bigm|
	JX+{}^tXJ=0\right\}.
\]
Also let $E_j$, $F_j$, $H_j$ $(j=0,\ldots,2n+2)$ be the Chevalley generators
for the loop algebra $\mathfrak{so}(4n+4)[z,z^{-1}]$ defined by
\[\begin{split}
	&E_0=zX_{4n+3,1},\quad E_i=X_{i,i+1},\quad E_{2n+2}=X_{2n+1,2n+3},\\
	&F_0=\frac{1}{z}X_{1,4n+3},\quad F_i=X_{i+1,i},\quad
	F_{2n+2}=X_{2n+3,2n+1},\\
	&H_0=-X_{1,1}-X_{2,2},\quad H_i=X_{i,i}-X_{i+1,i+1},\\
	&H_{2n+2}=X_{2n+1,2n+1}+X_{2n+2,2n+2},
\end{split}\]
for $i=1,\ldots,2n+1$, where $X_{i,j}=E_{i,j}-E_{4n+5-j,4n+5-i}$.
Note that
\[
	H_0 + H_1 + \sum_{i=2}^{2n}2H_i + H_{2n+1} + H_{2n+2} = 0.
\]

Under a specialization $K=0$, we can identify this loop algebra with the
affine Lie algebra $\mathfrak{g}(D^{(1)}_{2n+2})$.
Note that the scaling element $d$ corresponds to the differential operator
$z\partial_z$.
We also remark that
\[
	[X,Y] = XY-YX,\quad (X|Y) = \frac{1}{2}\mathrm{tr}XY.
\]

In a similar manner as \cite{DF}, we formulate the Heisenberg subalgebra
of type $(1,1,0,1,0,\ldots,1,0,1,1)$ in a flamework of
$\mathfrak{so}(4n+4)[z,z^{-1}]$.
Let $\Lambda_{1,i}$ $(i=1,2)$ be matricies defined by
\[\begin{split}
	\Lambda_{1,1} &= E_0 + [E_1,E_2]
	+ \sum_{j\in\mathcal{J}'}(E_j+[E_{j-1},[E_j,E_{j+1}]])
	+ E_{2n+1} + [E_{2n},E_{2n+2}],\\
	\Lambda_{1,2} &= E_1 + [E_0,E_2]
	+ \sum_{j\in\mathcal{J}'}([E_{j-1},E_j]+[E_j,E_{j+1}])
	+ E_{2n+2} + [E_{2n},E_{2n+1}].
\end{split}\]
Note that $[\Lambda_{1,1},\Lambda_{1,2}]=0$.
We also set
\[
	\Lambda_{(2n+2)k+l,i} = z^k(\Lambda_{1,i})^l\quad
	(i=1,2; k\in\mathbb{Z}; l=1,3,\ldots,2n+1).
\]
Then we have a maximal nilpotent subalgebra $\bigoplus_{k\in\mathbb{Z}}
\left(\mathbb{C}\Lambda_{2k-1,1}\oplus\mathbb{C}\Lambda_{2k-1,2}\right)$
of $\mathfrak{so}(4n+4)[z,z^{-1}]$.
It can be identified with the Heisenberg subalgebra $\mathfrak{s}$ given
in Section \ref{Sec:Lie_alg} under the specialization $K=0$.

\begin{rem}
The isomorphism classes of the Heisenberg subalgebras are in one-to-one
correspondence with the conjugacy classes of the finite Weyl group
{\rm\cite{KP}}.
In the notation of {\rm\cite{C}}, the Heisenberg subalgebra $\mathfrak{s}$
introduced above corresponds to the regular primitive conjugacy class
$D_{2n+2}(a_n)$ of the Weyl group of type $D_{2n+2}$.
\end{rem}

\section{Lax pair}\label{Sec:Lax}

It is known that $P_{\rm{VI}}$ is derived from the Lax pair associated
with the loop algebra $\mathfrak{so}(8)[z,z^{-1}]$ \cite{NY3}.
In this section, we propose a Lax pair for the system \eqref{Sys:Sasano}
with \eqref{Ham:Sasano} in a flamework of $\mathfrak{so}(4n+4)[z,z^{-1}]$.

In the previous section, we have derived the system \eqref{Sim_Red_ODE}.
It can be identified with the system on
$\mathfrak{so}(4n+4)[z,z^{-1}]$
\begin{equation}\label{Sim_Red_ODE_Mat}
	\left[s(s-1)\frac{d}{ds}-B,z\frac{z}{dz}+M\right] = 0,
\end{equation}
where
\[\begin{split}
	M &= \sum_{i=0}^{2n+2}\vep_iH_i + \sum_{i=0}^{2n+2}\varphi_iE_i
	+ \sum_{i=1}^{2n}\frac{[E_i,E_{i+1}]}{2n+2}
	+ \frac{[E_{2n},E_{2n+2}]}{2n+2},\\
	B &= \sum_{i=0}^{2n+2}u_iH_i + \sum_{i=0}^{2n+2}x_iE_i + y_1[E_0,E_2]
	+ \sum_{i=2}^{2n}y_i[E_i,E_{i+1}]\\
	&\quad + y_{2n+1}[E_{2n},E_{2n+2}]
	+ \sum_{j\in\mathcal{J}'}y_1[E_{j-1},[E_j,E_{j+1}]],
\end{split}\]
under the specialization $K=0$.
Here $\vep_i$ $(i=0,\ldots,2n+2)$ are complex parameters such as
\[\begin{split}
	&\alpha_0 = 1+2\vep_0-\vep_2,\quad \alpha_1 = 2\vep_1-\vep_2,\quad
	\alpha_2 = -\vep_0-\vep_1+2\vep_2-\vep_3,\\
	&\alpha_i = -\vep_{i-1}+2\vep_i-\vep_{i+1}\quad (i=3,\ldots,2n-1),\\
	&\alpha_{2n} = -\vep_{2n-1}+2\vep_{2n}-\vep_{2n+1}-\vep_{2n+2},\\
	&\alpha_{2n+1} = -\vep_{2n}+2\vep_{2n+1},\quad
	\alpha_{2n+2} = -\vep_{2n}+2\vep_{2n+2}.
\end{split}\]

Consider a system of linear differential equations
\begin{equation}\label{Lax}
	s(s-1)\frac{d\bm{w}}{ds} = B\bm{w},\quad
	z\frac{d\bm{w}}{dz} + M\bm{w} = 0,
\end{equation}
for a vector of unknown functions $\bm{w}={}^t(w_1,\ldots,w_{4n+4})$.
Then the system \eqref{Sim_Red_ODE_Mat} can be regarded as the
compatibility condition of \eqref{Lax}.
In this flamework, the group of symmetries \eqref{Aff_Wey_Sym} arise from
gauge transformations
\[
	r_i(\bm{w}) = \left(1+\frac{\alpha_i}{\varphi_i}F_i\right)\bm{w}\quad
	(i=0,\ldots,2n+2).
\]
Note that the Lax pair \eqref{Lax} of the case $n=1$ is equivalent to
one of \cite{NY3}.

The Lax pair \eqref{Lax} arises from the Drinfeld-Sokolov hierarchy
as follows.
Under the system \eqref{Eq:Sato}, we consider a $G_{<0}G_{\geq0}$-function
$\Psi=\Psi(t_1,t_2,\ldots)$ defined by
\[
	\Psi = W\exp\left(\sum_{k=1,2,\ldots}t_k\Lambda_k\right).
\]
Then we obtain
\begin{equation}\label{Lax_DS}
	\Psi\partial_{t_k}\Psi^{-1} = \partial_{t_k}-B_k\quad
	(k=1,2,\ldots),\quad
	\Psi\vartheta\Psi^{-1} = \mathcal{M}.
\end{equation}
Note that the system \eqref{ZS_DS} can be regarded as the compatibility
condition of \eqref{Lax_DS}.
In the following, we use a conventional form of \eqref{Lax_DS}
\[
	\partial_{t_k}(\Psi) = B_k\Psi\quad (k=1,2,\ldots),\quad
	\vartheta(\Psi) = (\vartheta-\mathcal{M})\Psi.
\]
It is equivalent to
\begin{equation}\label{Lax_Sim_Red}
	\partial_{t_k}(\Psi) = B_k\Psi\quad (k=1,2),\quad
	\vartheta(\Psi) = (t_1B_1+t_2B_2)\Psi,
\end{equation}
under the specialization $\mathcal{M}\in\mathfrak{g}_{\geq0}$ and
$t_k=0$ $(k\geq3)$.
Via a gauge transformation $\Psi^{+}=\exp(\Gamma)\Psi$, the system
\eqref{Lax_Sim_Red} is transformed into
\begin{equation}\label{Lax_Borel}
	\partial_{s_k}(\Psi^{+}) = B^{+}_k\Psi^{+}\quad (k=1,2),\quad
	\vartheta(\Psi^{+}) = (\vartheta-\mathcal{M}^{+})\Psi^{+},
\end{equation}
where $B^{+}_k$ $(k=1,2)$ are defined by
$\mathcal{B}^{+}=B^{+}_1ds_1+B^{+}_2ds_2$ and
\[
	s_1 = \frac{c_{2n+2}-c_1}{2n+2},\quad
	s_2 = \frac{c_{2n+2}-c_{2n+1}}{2n+2}.
\]
The system \eqref{Lax_Borel} can be identified with \eqref{Lax} under
the specialization $s_2=1$ and $K=0$.

\section*{Acknowledgement}
The authers are grateful to Professors Masatoshi Noumi, Yasuhiko Yamada and
Yusuke Sasano for valuable discussions and advices.



\begin{thebibliography}{9}
\bibitem[AS]{AS}
	M. J. Ablowitz and H. Segur,
	Exact linearization of a Painlev\'{e} transcendent,
	Phys. Rev. Lett. \textbf{38} (1977), 1103-1106.
\bibitem[C]{C}
	R. Carter,
	Conjugacy classes in the Weyl group,
	Compositio Math. \textbf{25} (1972), 1-59.
\bibitem[DF]{DF}
	F. Delduc and L. Feh\'{e}r,
	Regular conjugacy classes in the Weyl group and integral hierarchies,
	J. Phys. A: Math. Gen. \textbf{28} (1995), 5843-5882.
\bibitem[DS]{DS}
	V. G. Drinfel'd and V. V. Sokolov,
	Lie algebras and equations of Korteweg-de Vries type,
	J. Sov. Math. \textbf{30} (1985), 1975-2036.
\bibitem[FS]{FS}
	K. Fuji and T. Suzuki,
	The sixth Painlev\'{e} equation arising from $D^{(1)}_4$ hierarchy,
	J. Phys. A: Math. Gen., \textbf{39} (2006) 12073-12082.
\bibitem[IKSY]{IKSY}
	K. Iwasaki, H. Kimura, S. Shimomura and M. Yoshida,
	From Gauss to Painlev\'{e} --- A Modern Theory of Special Functions,
	Aspects of Mathematics \textbf{E16} (Vieweg, 1991).
\bibitem[Kac]{Kac}
	V. G. Kac,
	Infinite dimensional Lie algebras,
	Cambridge University Press (1990).
\bibitem[KIK]{KIK}
	T. Kikuchi, T. Ikeda and S. Kakei,
	Similarity reduction of the modified Yajima-Oikawa equation,
	J. Phys. A: Math. Gen., \textbf{36} (2003) 11465-11480.
\bibitem[KK1]{KK1}
	S. Kakei and T. Kikuchi,
	Affine Lie group approach to a derivative nonlinear Schr\"{o}dinger
	equation and its similarity reduction,
	Int. Math. Res. Not. \textbf{78} (2004), 4181-4209.
\bibitem[KK2]{KK2}
	S. Kakei and T. Kikuchi,
	The sixth Painlev\'{e} equation as similarity reduction of
	$\wh{\mathfrak{gl}}_3$ hierarchy,
	preprint (nlin-SI/0508021).
\bibitem[KP]{KP}
	V. G. Kac and D. Peterson,
	112 constructions of the basic representation of the roop group of $E_8$,
	in Symposium on Anomalies, Geometry ans Topology, ed. W. A. Baedeen and
	A. R. White, (World Scientific, 1985) 276-298.
\bibitem[NY1]{NY1}
	M. Noumi and Y. Yamada,
	Higher order Painlev\'{e} equations of type $A^{(1)}_l$,
	Funkcial. Ekvac. \textbf{41} (1998), 483-503.
\bibitem[NY2]{NY2}
	M. Noumi and Y. Yamada,
	Birational Weyl group action arising from a nilpotent Poisson algebra,
	in Physics and Combinatorics 1999, Proceedings of the Nagoya 1999
	International Workshop, ed. A.N.Kirillov, A.Tsuchiya and H.Umemura,
	(World Scientific, 2001) 287-319.
\bibitem[NY3]{NY3}
	M. Noumi and Y. Yamada,
	A new Lax pair for the sixth Painlev\'{e} equation associated with
	$\wh{\mathfrak{so}}(8)$,
	in Microlocal Analysis and Complex Fourier Analysis, ed. T.Kawai and
	K.Fujita, (World Scientific, 2002) 238-252.
\bibitem[O]{O}
	K. Okamoto,
	Studies on the Painlev\'{e} equations,
	I, Ann. Math. Pura Appl. \textbf{146} (1987), 337--381,
	II, Jap. J. Math. \textbf{13} (1987), 47--76,
	III, Math. Ann. \textbf{275} (1986), 221--256,
	IV, Funkcial. Ekvac. \textbf{30} (1987), 305--332.
\bibitem[S]{S}
	Y. Sasano,
	Higher order Painlev\'{e} equations of type $D^{(1)}_l$,
	RIMS Koukyuroku \textbf{1473} (2006) 143-163.
\end{thebibliography}
\end{document}